\documentclass[11pt,hyperref]{amsart}

\usepackage[svgnames]{xcolor}
\usepackage{amssymb}
\usepackage[colorlinks=true,citecolor=red!60!black]{hyperref}
\usepackage{amsmath}
\usepackage{amsthm}
\usepackage[mathcal]{euscript}
\usepackage[T1]{fontenc}
\usepackage[latin1]{inputenc}
\usepackage{ae,aecompl}
\usepackage{enumitem}
\usepackage{float}
\usepackage{xspace}
\usepackage{microtype}
\usepackage{pdfsync}
\usepackage[active]{srcltx}

\usepackage[textwidth=4cm]{todonotes}
\newcommand{\mynote}[3][]{\todo[caption={\sf #3}, color={%
    \ifnum#2=0 green!20
    \else\ifnum#2=1 orange!30
    \else\ifnum#2=2 yellow!20
    \else\ifnum#2=3 cyan!20
    \else magenta!20\fi\fi\fi\fi}, size=\tiny, #1]{\renewcommand{\baselinestretch}{1}\selectfont\sf#3}\xspace}

\let\leq\leqslant
\let\geq\geqslant

\usepackage[pdflatex=true,recompilepics=auto]{gastex}

\setenumerate[1]{label=$(\arabic*)$,leftmargin=*}
\setenumerate[2]{label=$(\alph*)$,leftmargin=*}
\setenumerate[3]{label=$(\roman*)$,leftmargin=*}
\setitemize[1]{label=$-$,leftmargin=*}
\setitemize[2]{label=$\circ$,leftmargin=*}
\setitemize[3]{label=$--$,leftmargin=*}

\newtheorem{myclaim}{Claim}
\newtheorem{myremark}{Remark}
\newtheorem{lemma}{Lemma}
\newtheorem*{lemma*}{Lemma}
\newtheorem{theorem}{Theorem}
\newtheorem{proposition}{Proposition}
\newtheorem{corollary}{Corollary}

\DeclareSymbolFont{bbold}{U}{bbold}{m}{n}
\DeclareSymbolFontAlphabet{\mathbbold}{bbold}

\newcommand\content[1]{\ensuremath{\contentmorphism(#1)}}
\newcommand\contentscc[2]{\ensuremath{\contentmorphism\_\mathsf{scc}(#1,#2)}}
\newcommand\arity[1]{\ensuremath{\textsf{arity}(#1)}}
\newcommand\contentmorphism{\ensuremath{\textsf{alph}}}
\newcommand\scc[2]{\ensuremath{\textsf{scc}(#1,#2)}}
\newcommand\PT{PT}
\newcommand\exactcontent[1]{\ensuremath{{#1}^\circledast}}
\newcommand\LangExp[4]{{#1}_0(\exactcontent{#2}_1)^{#3}{#1}_1\cdots {#1}_{{#4}-1}(\exactcontent{#2}_{#4})^{#3}{#1}_{#4}}
\newcommand\Lang[3]{{\ensuremath{\mathcal{L}}(\vec {#1},\vec {#2},#3)}}

\def\pv#1{\ensuremath{{\sf#1}}}

\def\Om#1#2{\ensuremath{\overline\Omega_{#1}{\sf#2}}}
\def\om#1#2{\ensuremath{\Omega_{#1}{\sf#2}}}

\def\tcl#1{\ensuremath{\mathrm{cl}(#1)}}

\title[The separation problem for regular languages]{The separation problem for regular languages by piecewise testable languages}

\author[L. van Rooijen]{L. van Rooijen} 
\email{lvanrooi@labri.fr}
\author[M. Zeitoun]{M. Zeitoun}
\email{mz@labri.fr}

\address{LaBRI, Universit\'es de Bordeaux \& CNRS
  UMR~5800. 351 cours de la Lib\'eration, 33405 Talence Cedex, France.
}


\begin{document}

\begin{abstract}
Separation is a classical problem in mathematics and computer science. It asks whether, given two sets belonging to some class, it is possible to separate them by another set of a smaller class. We present and discuss the separation problem for regular languages. We then give a direct polynomial time algorithm to check whether two given regular languages are separable by a piecewise testable language, that is, whether a $\mathcal{B}\Sigma_1(<)$ sentence can witness that the languages are indeed disjoint. The proof is a reformulation and a refinement of an algebraic argument already given by Almeida and the second author.
\end{abstract}

\maketitle

\section{Introduction}
\label{sec:introduction}

\subsection*{The separation problem} Separation is a classical question in mathematics and computer science. In general, one says that two sets $X,Y$ are \emph{separable} by a set $U$ if $X\subseteq U$ and $Y\cap U=\varnothing$. In this case, $U$ is called a \emph{separator}.

The separation problem is the following. Consider a class $\mathcal{C}$ of sets or structures, and a subclass $\mathcal{C}_0$ of $\mathcal{C}$. The problem asks whether two elements $X,Y$ of $\mathcal{C}$ can always be separated by an element of the subclass $\mathcal{C}_0$. A classical example of such a separation problem, with a positive answer, is the Hahn-Banach separation theorem. Another example that appeared recently in computer science is the 
proof of Leroux~\cite{LEROUX-TURING100} of the decidability of the reachability problem for vector addition systems (or Petri Nets), which greatly simplifies the original proof by Mayr~\cite{DBLP:journals/siamcomp/Mayr84}, and that of Kosaraju~\cite{Kosaraju:1982:DRV:800070.802201}.  Namely, Leroux has shown that non-reachability can be witnessed by a class of recursively enumerable separators: from a configuration $c_1$ of such a system, one cannot reach a configuration $c_2$ if and only if the sets $\{c_1\}$ and $\{c_2\}$ can be separated by a Presburger definable set, which in addition is invariant under actions of the vector addition system. Since such sets form a recursively enumerable class, this yields a semi-algorithm for checking non-reachability.

In the case where elements of $\mathcal{C}$ cannot always be separated
by an element of $\mathcal{C}_0$, several natural questions arise:
\begin{enumerate}
\item\label{item:1} given elements $X,Y$ in $\mathcal{C}$, can we decide whether a separator exists in~$\mathcal{C}_0$?
\item if so, what is the complexity of this decision problem?
\item can we in addition compute a separator, and what is the complexity?
\end{enumerate}

In this context, it is known for example that separation of two  context-free languages by a regular one is undecidable~\cite{Hunt:1982:DGP:322307.322317}.

In this paper, we look at the separation problem for the class $\mathcal{C}$ of regular languages, and we are looking for separators in smaller classes, such as prefix- or suffix-testable languages, locally trivial languages, and piecewise testable languages (we will define these classes below).

\subsection*{The profinite approach} Several results from the literature can be combined into an algorithm answering question~\ref{item:1}, for all classes we are interested in. Several partial complexity results can also be derived from this approach, which we briefly explain now. This approach relies on a generic connection found by Almeida~\cite{MR1709911} between profinite semigroup theory and the separation problem, when the separators are required to belong to a given variety of regular languages.

\smallskip
A variety $\mathcal{V}$ of regular languages associates to each finite alphabet $A$ a class of languages $A^*\mathcal{V}$, with some closure 
properties (namely closure under Boolean operations, left and right residuals $L\mapsto a^{-1}L$ and $L\mapsto La^{-1}$, and inverse morphisms between free monoids). All classes of separators in this paper belong to a variety of regular languages.

Almeida~\cite{MR1709911} has shown that two regular languages over $A$ are separable by a language of $A^*\mathcal{V}$ if and only if the topological closures of these two languages inside a profinite semigroup, depending only on $\mathcal{V}$, intersect. To turn this property into an algorithm, we have therefore to be able:
\begin{itemize}
\item to compute representations of these topological closures, and
\item to test for emptiness of intersections of such closures.
\end{itemize}
So far, these problems have no generic answer. They have been studied for a small number of specific varieties, in an algebraic context. Deciding whether the closures of two regular languages intersect is equivalent to computing the so-called 2-pointlike sets of a finite semigroup wrt.~the variety we are interested in, see~\cite{MR1709911}. This question has been answered positively, in particular for the following varieties:
\begin{enumerate}[label=$\roman*)$]
\item languages recognized by a finite group~\cite{ash:1991:a,ribes&zalesskii:1993:a,Auinger&Steinberg:constructive-version-Ribes-Zalesskii-product:2005:a},
\item star-free (that is, FO-definable) languages~\cite{DBLP:journals/ijac/HenckellRS10a,Henckell:1988},
\item\label{item:4} piecewise testable (that is, $\mathcal{B}\Sigma_1(<)$-definable) languages~\cite{MR1611659,MR2365328},
\item languages whose syntactic semigroups are $\mathcal{R}$-trivial, that is, languages whose minimal automaton is very weak (the only cycles allowed in the graph of the automaton are self-loops)~\cite{MR2365328},
\item\label{item:5} languages for which membership can be tested by inspecting prefixes and suffixes up to some length (folklore, see~\cite[Sec.~3.7]{JAbook}),
\item locally testable languages, that is, languages for which membership can be tested by inspecting prefix, suffix and factors up to some length~\cite{MR1851812,Nogueira}.
\end{enumerate}

For all these classes, proofs use algebraic or topological arguments. In this paper, we obtain direct polynomial time algorithms for Cases~\ref{item:4} and \ref{item:5}. Our intuition is strongly lead by the proof techniques from profinite semigroup theory. 

\medskip
A general issue is that the topological closures cannot be described with a finite device. However, for 
piecewise testable languages, the approach of~\cite{MR1611659} consists in computing an automaton over an extended alphabet, which recognizes the closure of the original language. This can be performed in polynomial time wrt.\ the size of the original automaton. Since these automata admit the usual construction for intersection, and can be checked for emptiness in NLOGSPACE, we get a polynomial time algorithm wrt.~the size of the original automata. The construction was presented for deterministic automata but also works for nondeterministic ones. One should mention that the extended alphabet is $2^A$ (where $A$ is the original alphabet).
Therefore, these results give an algorithm which, from two NFAs, decides separability by piecewise testable languages in time polynomial in the number of states of the NFAs and exponential in the size of the original alphabet.


The improvement of the separation result for piecewise testable languages as presented in this paper is twofold: on the one hand, the algorithm presented provides better complexity as it runs in polynomial time in both the size of the automata, \emph{and} in the size of the alphabet. On the other hand, our results do not make use of the theory of profinite semigroups, that is, we work only with elementary concepts. The proof follows however basically the  same pattern as the original one.

The key argument is to show that non-separability is witnessed by both automata admitting a path of the same shape. 
In our proof, we manually extract from two non-separable automata some paths with this property, using Simon's factorization forest Theorem~\cite{Simon199065}. Whereas in the profinite world, these witnesses are  immediately obtained by a standard compactness argument.

\subsection*{Organization of the paper} After having recalled the background in Section~\ref{sec:prelim}, we present in Section~\ref{sec:separation-prefixes} a simple toy example, to highlight the main definitions and techniques: the case of separation by prefix-testable languages. Section~\ref{sec:separation} is devoted to the question of separation by piecewise testable languages. The main algorithm and proofs are given in this section. For the interested reader, we provide some elements of profinite semigroup theory in appendix.


\section{Preliminaries}
\label{sec:prelim}
Given a finite alphabet $A$, we denote by $A^*$ (resp.~by $A^+$) the free monoid (resp.~the free semigroup) over $A$. For a word $u \in A^*$, the smallest $B\subseteq A$ such that $u\in B^*$ is called the \emph{alphabet} of $u$ and is denoted by $\content u$. A \emph{nondeterministic finite automaton} (NFA) over $A$ is denoted by a tuple $\mathcal{A}=(Q,A,I,F,\delta)$, where $Q$ is the set of states, $I\subseteq Q$ the set of initial states, $F\subseteq Q$ the set of final states and $\delta\subseteq Q\times Q$ the transition relation. If $\delta$ is a function, then $\mathcal{A}$ is a deterministic automaton (DFA). We denote by $L(\mathcal{A})$ the language of words accepted by $\mathcal{A}$.  Given a word $u\in A^*$, a subset $B$ of $A$ and two states $p,q$ of $\mathcal{A}$, we denote 
\begin{itemize}
\item by $p\xrightarrow{\ u\ }q$ a path from state $p$ to state $q$ labeled $u$. 
\item by $p\xrightarrow{{}\subseteq B} q$ a path from $p$ to $q$ of which all   transitions are labeled by letters of~$B$.
\item by $p\xrightarrow{{}=B}q$ a path from $p$ to $q$ of which all transitions are labeled by letters of   $B$, with the additional demand that every letter of $B$ occurs at   least once along this path.
\end{itemize}
Given a state $p$, we denote by $\scc{p}{\mathcal{A}}$ the strongly connected component of $p$ in $\mathcal{A}$ (that is, the set of states reachable from $p$), and by $\contentscc{p}{\mathcal{A}}$ the set of labels of all transitions occurring in this strongly connected component. Finally, we define the restriction of $\mathcal{A}$ to a subalphabet  $B\subseteq A$ by  $\mathcal{A}\restriction_B\mathrel{\;\stackrel{\text{def}}=}(Q,A,I,F,\delta\cap (Q\times B\times Q))$.

\section{A toy example: separation by prefix-testable languages}
\label{sec:separation-prefixes}
A regular language $L$ is a \emph{prefix-testable language} if membership of $L$ can be tested by inspecting prefixes up to some length, that is, if $L$ is a finite Boolean combination of languages of the form $uA^*$, for a finite word $u$. Prefix-testable languages form a variety of regular languages. Therefore, as recalled in the introduction, it follows by~\cite{MR1709911} that testing whether two given languages can be separated by a prefix-testable language amounts to checking that their topological closures in some profinite semigroup have a nonempty intersection.

\smallskip
 It turns out that for prefix-testable languages, this profinite semigroup is easy to describe (see~\cite[Sec.~3.7]{JAbook}): it is $A^+\cup A^\infty$, where $A^\infty$ denotes the set of right infinite words over~$A$. Multiplication in this semigroup is defined as follows: infinite words are left zeros ($vw=v$ if $v\in A^\infty$), and  multiplication on the left by a finite word is the usual multiplication: $(a_1\cdots a_n)(b_1\cdots)=a_1\cdots a_nb_1\cdots$. Finally, the topology is the product topology: a sequence converges
 \begin{itemize}
 \item to a finite word $u$ if it is ultimately equal to $u$,
 \item to an infinite word $v$ if for every finite prefix $x$ of $v$, the sequence ultimately belongs to $x(A^+\cup A^\infty)$.
 \end{itemize}

Therefore, from a given NFA $\mathcal{A}$, one can compute a Büchi automaton recognizing the language of infinite words that belong to the closure of~$L(\mathcal{A})$, as follows:
\begin{enumerate}
\item Trim $\mathcal{A}$, by removing all states from which one cannot reach a final state. This can be performed in linear time wrt.~the size of $\mathcal{A}$, and does not change the language recognized by $\mathcal{A}$.
\item Build the Büchi automaton obtained from the resulting trim automaton by declaring all states accepting. 
\end{enumerate}

This yields a straightforward PTIME (actually NLOGSPACE) algorithm to decide separability by a prefix-testable language: first check that $L(\mathcal{A}_1)\cap L(\mathcal{A}_2)=\varnothing$. If so, compute the intersection of the languages of infinite words belonging to the closures of $L(\mathcal{A}_1)$ and $L(\mathcal{A}_2)$ by the usual product construction, and check that this Büchi automaton accepts at least one word. 

\begin{proposition}
  \label{prefix-testable}
  One can decide in PTIME whether two languages can be separated by a prefix-testable language.\qed
\end{proposition}


\section{A simple PTIME algorithm for separation by a piecewise testable language}
\label{sec:separation}

\subsection*{Piecewise testable languages} Let $\lhd$ be the \emph{scattered subword ordering} defined on $A^*$ as follows: for $u,v\in A^*$, we have $u\lhd v$ if $u=a_1\cdots a_n$ and $v=v_0a_1v_1\cdots v_{n-1}a_nv_n$, with $a_i\in A$ and  $v_i\in A^*$.
We let $$\text{Sub}_n(u)=\{ w \in A^* : |w| \leq n, w \lhd u\}.$$
When two words have the same scattered subwords up to length $n$, we say that they are \emph{$\sim_n$-equivalent}:
\[
\text{Sub}_n(u) = \text{Sub}_n(v)\quad \Longleftrightarrow\quad u \sim_n v.
\] 

A regular language over an alphabet $A$ is \emph{piecewise testable (\PT)} \cite{Simon:1975:PTE:646589.697341} if it is a finite Boolean combination of languages of the form $A^* a_1 A^* a_2 \ldots A^* a_n A^*$, where every $a_i \in A$.  Whether a given word belongs to a \PT-language is thus determined by the set of its scattered subwords up to a certain length. In other words, a regular language $L$ is piecewise testable if and only if there exists an $n \in \mathbb{N}$ such that $L$ is a union of $\sim_n$-classes. 

\smallskip
The class of piecewise testable languages has been extensively studied during the last decades. It corresponds to languages that can be defined in the fragment $\mathcal{B}\Sigma_1(<)$ of first-order logic on finite words. Simon has shown that piecewise testable languages are exactly those languages whose syntactic monoid is $\mathcal{J}$-trivial~\cite{Simon:1975:PTE:646589.697341}, and this property yields a decision procedure to check  whether a language is piecewise testable. Stern has refined this procedure into a polynomial time algorithm~\cite{Stern:Complexity-some-problems-from:1985:a}, whose complexity has been improved by Trahtman~\cite{Trahtman:Piecewise-Local-Threshold-Testability:2001:b}.

\subsection*{Separation by a piecewise testable language} We say that two regular languages $L_1, L_2$ are \emph{\PT-separable} if there exists a piecewise testable language $L$ that separates them, \emph{i.e.},
\begin{equation*}
L_1 \subseteq L \text{ and } L_2 \cap L = \varnothing.
\end{equation*}
 In other words, $L_1$ and $L_2$ are \PT-separable if there exists a $\mathcal{B}\Sigma_1(<)$ formula which is satisfied by all words of $L_1$, and not satisfied by any  word of $L_2$.

Our main contribution is a simple proof of the following result, which  states that one can decide in polynomial time whether two languages are \PT-separable.

\begin{theorem}
  \label{thm:PT-separation}
  Given two NFAs, one can determine in polynomial time, with respect to the number of states and the size of the alphabet, whether the languages recognized by these NFAs are \PT-separable.
\end{theorem}

Note that a language is \PT-separable from its complement if and only if it is piecewise testable itself. Therefore, applying Theorem~\ref{thm:PT-separation} to a language and to its complement if they are both given by NFAs yields a polynomial time algorithm to check if a language is piecewise testable. We recover in particular the following result, proved by Stern~\cite{Stern:Complexity-some-problems-from:1985:a} using the characterization for minimal automata recognizing \PT-languages as given by Simon in~\cite{Simon:1975:PTE:646589.697341} (this result has later been improved by Trahtman~\cite{Trahtman:Piecewise-Local-Threshold-Testability:2001:b}). 
\begin{corollary}
  \label{cor:pt-decidability}
  One can decide in polynomial time whether a given DFA recognizes a  piecewise testable language.
\end{corollary}

The rest of this section is devoted to the proof of Theorem~\ref{thm:PT-separation}.  We fix a DFA $\mathcal{A}$ over~$A$. For $u_0,\ldots,u_p\in A^*$ and \underbar{nonempty} subalphabets $B_1,\ldots,B_p\subseteq A$, let $\vec u=(u_0,\ldots,u_p)$ and $\vec B=(B_1,\ldots,B_p)$.  We call such a pair $(\vec u, \vec B)$ a \emph{factorization pattern}.  A \emph{$(\vec u,\vec B)$-path} in $\mathcal{A}$ is a successful path (leading from the initial state to a final state of $\mathcal{A}$), of the form
\begin{figure}[H]
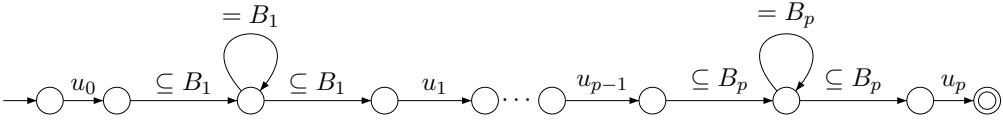

  \hspace*{-60ex}
  \vspace*{7ex}
  \centering
  \scalebox{.89}{%
    \begin{gpicture}
      \gasset{Nw=4,Nh=4}

      \put(0,-11){%
        \node[Nmarks=i](0)(-10,0){}
      \node(1)(0,0){}
      \node(2)(20,0){}
      \node(3)(40,0){}
      \node(4)(55,0){}
      \drawedge(0,1){$u_0$}
      \drawedge(1,2){$\subseteq B_1$}
      \drawedge(2,3){$\subseteq B_1$}
      \drawloop(2){$=B_1$}
      \drawedge(3,4){$u_1$}
      \node[Nframe=n](x)(60,0){$\cdots$}
      \put(80,0){\node(0)(-15,0){}
        \node(1)(0,0){}
        \node(2)(20,0){}
        \node(3)(40,0){}
        \node[Nmarks=r](4)(50,0){}
        \drawedge(0,1){$u_{p-1}$}
        \drawedge(1,2){$\subseteq B_p$}
        \drawedge(2,3){$\subseteq B_p$}
        \drawloop(2){$=B_p$}
        \drawedge(3,4){$u_p$}}
      }
    \end{gpicture}}
  \caption{A $(\vec u, \vec B)$-path}
  \label{fig:ub-path}
\end{figure}

Recall that edges denote sequences of transitions: an edge labeled ${}\subseteq B$ denotes a path of which all transitions are labeled by letters of $B$. An edge labeled ${}= B$ denotes a path of which all transitions are labeled by letters of $B$,  with the additional demand that every letter of $B$ occurs at least once.

\begin{myremark} \label{rem:sublang}
  The automaton $\mathcal{A}$ admits a $(\vec u,\vec B)$-path if and only if $L(\mathcal{A})$ contains a language of the form
  \begin{equation*}
    u_0(x_1y_1^*z_1)u_1\cdots u_{p-1}(x_py_p^*z_p)u_p,
  \end{equation*}
  where $\content{x_i}\cup\content{z_i}\subseteq \content{y_i}=B_i$.
\end{myremark}

Theorem~\ref{thm:PT-separation} directly follows from the next two statements.

\begin{proposition}
  \label{prop:J-sep-criterion}
  Let $\mathcal{A}_1$ and $\mathcal{A}_2$ be two NFAs. Then,
  $L(\mathcal{A}_1)$ and $L(\mathcal{A}_2)$ are \emph{not} \PT-separable if and only if
  there exist $\vec u=(u_0,\ldots,u_p)$ and $\vec B=(B_1,\ldots,B_p)$
  such that both $\mathcal{A}_1$ and $\mathcal{A}_2$ both have a $(\vec u,\vec
  B)$-path.
\end{proposition}

\begin{proposition}
  \label{prop:common-ub-path}
  Given two NFAs, one can determine in polynomial time, with respect to the number of states and the size of the alphabet, whether there   exist $\vec u=(u_0,\ldots,u_n)$ and $\vec B=(B_1,\ldots,B_n)$ such   that both NFAs admit a $(\vec u,\vec B)$-path.
\end{proposition}

As observed above, the characterization of \PT-separable languages given in Proposition~\ref{prop:J-sep-criterion} can be applied to the minimal automata of a regular language and of its complement, to obtain a characterization for minimal automata recognizing \PT-languages. It turns out that with this approach, we retrieve exactly the same characterization as given by Simon in~\cite{Simon:1975:PTE:646589.697341}.

\smallskip
Let us first prove Proposition~\ref{prop:common-ub-path}. 

\begin{proof}[Proof of Prop.~\ref{prop:common-ub-path}]
  We will first show that the following problem is in PTIME: given   states $p_1,q_1,r_1$ in automaton $\mathcal{A}_1$ and $p_2,q_2,r_2$ in   automaton $\mathcal{A}_2$, determine whether there exists a nonempty   alphabet $B\subseteq A$ such that there is an $(= B)$-loop around both   $q_1$ and $q_2$, and $(\subseteq B)$-paths from $p_1$ to $r_1$ via   $q_1$ in $\mathcal{A}_1$,   and from $p_2$ to $r_2$ via $q_2$ in $\mathcal{A}_2$, as pictured in~\figurename~\ref{fig:ub-path2}.

  \begin{figure}[H]
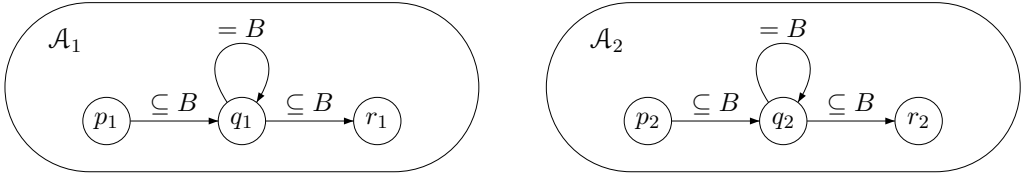

    \vspace*{17ex}
    \hspace*{-62ex}
    \scalebox{.9}{
      \begin{gpicture}
        \gasset{Nw=7,Nh=7}
        \def\aut#1#2#3#4{%
          \drawoval(20,5,70,25,100)
          \node[Nframe=n](A1)(-6,12){\large$\mathcal{A}_{#4}$}
          \node(1)(0,0){${#1}_{#4}$}
          \node(2)(20,0){${#2}_{#4}$}
          \node(3)(40,0){${#3}_{#4}$}
          \drawedge(1,2){$\subseteq B$}
          \drawedge(2,3){$\subseteq B$}
          \drawloop(2){$=B$}}
        \put(0,9){%
          \aut p q r 1
          \put(80,0){\aut p q r 2}}
      \end{gpicture}}
    \caption{Finding a common pattern in $\mathcal{A}_1$ and $\mathcal{A}_2$}
    \label{fig:ub-path2}
  \end{figure}

  To do so, we compute a decreasing sequence $(C_i)_i$ of alphabets over-approximating the maximal alphabet $B$ labeling the loops.    Note that if there exists such an alphabet $B$, it should be contained in
  $$C_1\stackrel{\text{def}}{=}\contentscc{q_1}{\mathcal{A}_1} \cap \contentscc{q_2}{\mathcal{A}_2}.$$
  Using Tarjan's algorithm to compute  strongly connected components in linear time~\cite{Cormen2001introduction},
  one can compute~$C_1$ in linear time as well. Then, we restrict the automata to alphabet $C_1$, and we repeat the process to obtain the sequence~$(C_i)_i$:
  \[
  C_{i+1} \stackrel{\text{def}}{=} \contentscc{q_1}{\mathcal{A}_1\restriction_{C_i}} \cap \contentscc{q_2}{\mathcal{A}_2\restriction_{C_i}}.
  \]
  After a finite number $n$ of iterations, we obtain $C_n =  C_{n+1}$. Note that $n\leq|\content{\mathcal{A}_1} \cap  \content{\mathcal{A}_2}|\leq |A|$. If $C_n = \varnothing$, then there exists  no nonempty $B$ for which there is an ($= B$)-loop around both $p$ and  $q$. If $C_n \ne \varnothing$, then it is the maximal nonempty  alphabet $B$ such that there are $(=B)$-loops around $q_1$ in  $\mathcal{A}_1$ and $q_2$ in $\mathcal{A}_2$. It then remains to determine  whether there exist paths $p_1\xrightarrow{{}\subseteq   B}q_1\xrightarrow{{}\subseteq B}r_1$ and $p_2\xrightarrow{{}\subseteq   B}q_2\xrightarrow{{}\subseteq B}r_2$, which can be performed in linear  time.

  To sum up, since the number $n$ of iterations to compute $C_n=C_{n+1}$ is bounded by $|A|$, and since each computation is linear wrt.~the size of   $\mathcal{A}_1$ and $\mathcal{A}_2$, deciding whether there is a pattern as in \figurename~\ref{fig:ub-path2} in both $\mathcal{A}_1$ and $\mathcal{A}_2$ can be done in polynomial time wrt.~to both $|A|$ and the size of the NFAs.

  Now we build from $\mathcal{A}_1$ and $\mathcal{A}_2$ two new automata   $\tilde{\mathcal{A}_1}$ and $\tilde{\mathcal{A}_2}$  as   follows. The procedure first initializes $\tilde{\mathcal{A}_i}$ as a   copy of $\mathcal{A}_i$. Denote by $Q_i$ the state set of $\mathcal{A}_i$. For each   4-uple $\tau=(p_1,r_1,p_2,r_2)\in Q_1^2\times Q_2^2$ such that there   exist an alphabet $B$, two states $q_1\in Q_1,q_2\in Q_2$ and paths   $p_i\xrightarrow{{}\subseteq     B}q_i\xrightarrow{{}=B}q_i\xrightarrow{{}\subseteq B}r_i$ both for   $i=1$ and $i=2$, we add in both $\tilde{\mathcal{A}_1}$ and $\tilde{\mathcal{A}_2}$     a new letter $a_\tau$ to the alphabet, and   transitions $p_1\xrightarrow{a_\tau}r_1$ and   $p_2\xrightarrow{a_\tau}r_2$. Since there is a   polynomial number of tuples $(p_1,q_1,r_1,p_2,q_2,r_2)$, the above shows that computing these new transitions can be performed in polynomial time. Therefore, computing $\tilde{\mathcal{A}_1}$ and   $\tilde{\mathcal{A}_2}$ can be done in PTIME.

  Now by construction, there exists some factorization pattern $(\vec u,   \vec B)$ such that $\mathcal{A}_1$ and $\mathcal{A}_2$ both have a   $(\vec u, \vec B)$-path if and only if $L(\tilde{\mathcal{A}_1})\cap   L(\tilde{\mathcal{A}_1})\not=\varnothing$. Since   both $\tilde{\mathcal{A}_1}$ and $\tilde{\mathcal{A}_1}$ have been   built in PTIME, this can be decided in polynomial time.
\end{proof}

As a side remark, let us mention that it is crucial that the (${}=B$)-paths, which are required to use exactly the same alphabets, are actually \emph{loops} (occurring in \figurename~\ref{fig:ub-path2} around states $q_1$ and $q_2$). The next statement shows that even for DFAs, the problem is NP-hard if we are looking for paths labeled by a common alphabet, without requesting these paths to be loops. The proof is deferred to the Appendix.
\begin{lemma}\label{lem:sat red}
  The following problem is NP-complete:
  \begin{tabbing}
    \textsf{Input:} \hspace{4mm} \= An alphabet $A = \{a_1, a_2, \ldots, a_n\}$ and two DFA's $\mathcal{A}_1, \mathcal{A}_2$ over $A$.\\
    
    \textsf{Question:} \> Do there exist $u \in L(\mathcal{A}_1)$ and $v \in L(\mathcal{A}_2)$ such that $\content{u} = \content{v}$?
  \end{tabbing}
\end{lemma}

Let us now prove Proposition~\ref{prop:J-sep-criterion}. Let us first prove the ``if'' direction. The ``only if'' direction is proved in Lemma~\ref{lem:common-ub-path-2}. 

\begin{lemma}
  \label{lem:common-ub-path-1}
  If two NFAs $\mathcal{A}_1$ and $\mathcal{A}_2$ share a common $(\vec u,\vec B)$ path, then the languages $L(\mathcal{A}_1)$ and $L(\mathcal{A}_2)$ are not \PT-separable.
\end{lemma}

\begin{proof}
  Let $L$ be a piecewise testable language such that $L(\mathcal{A}_1) \subseteq L$. Using the hypothesis and Remark \ref{rem:sublang}, this implies that $L$ contains a language $$u_0(x_1y_1^* z_1)u_1\cdots u_{p-1}(x_py_p^* z_p)u_p,$$ where $\content{x_i}\cup\content{z_i}\subseteq \content{y_i}=B_i$.
  Similarly, $L(\mathcal{A}_2)$ contains a language $u_0(x'_1y_{1}'^{*}z'_1)u_1\cdots u_{p-1}(x'_py_p'^*z'_p)u_p$, where $\content{x'_i}\cup\content{z'_i}\subseteq \content{y'_i}=B_i$. We will show that for every $n$, there is an element in this language which is $\sim_n$-equivalent to an element of $u_0(x_1y_1^*z_1)u_1\cdots u_{p-1}(x_py_p^*z_p)u_p$, using the following claim.
  \begin{myclaim}
    \label{claim:1}
    Given $x, x', y, y', z, z' \in A^*$ that satisfy 
    \[
    \begin{array}{ccl}
      \content{x} \cup \content{z} & \subseteq & \content{y}, \\ 
      \content{x'} \cup \content{z'} & \subseteq & \content{y'} = \content{y},
    \end{array}
    \]
    then for every $n \in \mathbb{N}$, 
    \[
    x y ^n z \sim_n x' y'^n z'. 
    \]
  \end{myclaim}
  \noindent
  Indeed, from the inclusions
  \begin{center}
    $\content{y} ^{\leq n} = $ Sub$_n (y^n) \subseteq $ Sub$_n (x y^n
    z) \subseteq \content{y} ^{\leq n}$,
  \end{center}
  it follows that Sub$_n (x y^n z) = \content{y} ^{\leq
    n}$. In the same way, Sub$_n (x'y'^nz') = \content{y'} ^{\leq n}$, which is
  equal to $\content{y} ^{\leq n}$. Thus $x y ^n z \sim_n x' y'^n
  z'$. This establishes the claim.

  Applying this, we obtain that $x_i y_i ^n z_i \sim_n x'_{i} y_{i}'^{n}
  z_i'$  for every $i$. Since $\sim_n$ is a congruence
  , we obtain for all $n \in \mathbb{N}$: 
  \[
  u_0(x_1y_1^n z_1)u_1\cdots u_{p-1}(x_py_p^n z_p)u_p \sim_n u_0(x'_1y_{1}'^n z'_1)u_1\cdots u_{p-1}(x'_py_p'^n z'_p)u_p.
  \]
  Since $L$ is piecewise testable, it is a union of $\sim_n$-equivalence classes for some $n$, thus it cannot be disjoint from $L(\mathcal{A}_2)$.  
\end{proof}

To prove the other direction of Proposition~\ref{prop:J-sep-criterion}, we introduce some notation. For $B\subseteq A$, let us denote by $\exactcontent B$ the set of words with alphabet exactly $B$: $$\exactcontent{B} = \{ w \in B^*\ |\ \content{w} = B\}.$$
Given a factorization pattern $(\vec u,\vec B)$, with $\vec
u=(u_0,\ldots,u_p)$ and $\vec B=(B_1,\ldots,B_p)$, we let
\begin{equation*}
  \Lang u B n=\LangExp  u B n p.
\end{equation*}
We say that a sequence $(w_n)_n$ is \emph{$(\vec u,\vec B)$-adequate} if
\begin{equation*}
  \forall n\geq0,\ w_n\in\Lang u B n.
\end{equation*}
A sequence is called \emph{adequate} if it is $(\vec u,\vec B)$-adequate for some factorization pattern~$(\vec u,\vec B)$.

\begin{lemma}\label{lem:extract-adequate-subseq}
  Every sequence $(w_n)_n$ of words admits an adequate subsequence.
\end{lemma}
\begin{proof}
  We use Simon's Factorization Forest Theorem, which we recall. See
  \cite{Simon199065,DBLP:conf/mfcs/Kufleitner08,DBLP:journals/tcs/Colcombet10}
  for proofs and extensions of this theorem. A
  \emph{factorization tree} of a nonempty word $x$ is a finite ordered unranked tree $T(x)$
  whose nodes are labeled by nonempty words, and such that:
  \begin{itemize}
  \item all leaves of $T(x)$ are labeled by letters,
  \item all internal nodes of $T(x)$ have at least 2 children,
  \item if a node labeled $y$ has $k$ children labeled $y_1,\ldots, y_k$
    from left to right, then $y=y_1\cdots y_k$.
  \end{itemize}
  Given a semigroup morphism $\varphi:A^+\to S$ into a finite semigroup
  $S$, such a factorization tree is \emph{$\varphi$-Ramseyan} if every
  internal node has either 2 children, or $k$ children labeled
  $y_1,\ldots, y_k$, in which case $\varphi$ maps all words $y_1,\ldots,y_k$ to the
  same idempotent of~$S$. Simon's Factorization Forest Theorem states that
  every word has a $\varphi$-Ramseyan factorization tree of height at most $3|S|$.

  Let now $(w_n)_n$ be a sequence of words. We use Simon's Factorization
  Forest Theorem with the morphism $\contentmorphism
  :A^+\to2^A$.
  
  Consider a sequence $(T(w_n))_n$, where $T(w_n)$ is an $\contentmorphism$-Ramseyan tree
  given by the Factorization Forest Theorem. In particular, 
  $T(w_n)$ has depth at most $3 \cdot 2^{|A|}$. Therefore, extracting a subsequence if
  necessary, one may assume that the sequence of depths of the trees
  $T(w_n)$ is a constant $H$. We argue by induction on $H$. If $H=0$, then
  every $w_n$ is a letter. Hence, one may extract from $(w_n)_n$
  a constant subsequence, which is therefore adequate.

  We denote the arity of the root of $T(w_n)$ by $\arity{w_n}$, and we
  call it the arity of $w_n$. If $H>0$, two cases may arise:
  \begin{enumerate}
  \item One can extract from $(w_n)_n$ a subsequence of bounded arity.
    Therefore, one may extract from $w_n$ a subsequence of constant
    arity, say $K$. This implies that each $w_n$ has a factorization in
    $K$ factors
    $$w_n=w_{n,1}\cdots w_{n,K},$$
    where $w_{n,i}$ is the label of the $i$-th child of the root in
    $T(w_n)$.  Therefore, the $\contentmorphism$-Ramseyan subtree of each
    $w_{n,i}$ is of height at most $H-1$. By induction, one can extract
    from $(w_{n,i})_n$ an adequate subsequence. Proceeding iteratively
    for $i=1,2,\ldots K$, one extracts from $(w_n)_n$ a subsequence
    $(w_{\sigma(n)})_n$ such that every $(w_{\sigma(n),i})_n$ is adequate.
    But a finite product of adequate sequences is obviously
    adequate. Therefore, the subsequence $(w_{\sigma(n)})_n$ of
    $(w_n)_n$ is also adequate.
    
  \item The arity of $w_n$ grows to infinity. Therefore, extracting if
    necessary, one can assume for every $n$, $\arity{w_n}\geq
    \max(n,3)$.  Since all arities of words in the sequence are at least
    3, all children of the root map to the same idempotent in
    $2^A$. But this says that each word from the subsequence is of the
    form
    $$w_{\sigma(n)}=w_{n,1}\cdots w_{n,K_n},$$
    with $K_n\geq n$, and where the alphabet of $w_{n,i}$ is the same
    for all $i$, say $B$. Therefore, $w_{\sigma(n)}\in (\exactcontent
    B)^{K_n}\subseteq(\exactcontent B)^n$. Therefore, $(w_{\sigma(n)})_n$ is adequate.\popQED    
  \end{enumerate}
\end{proof}

\smallskip\noindent
We now say that a factorization pattern $(\vec u,\vec B)$ is \emph{proper} if
\begin{enumerate}
\item\label{item:2} 
  for all $i$, $\text{last}(u_i) \notin B_i$ and $\text{first}(u_i) \notin B_{i-1}$, 
\item\label{item:3} 
  for all $i$, $u_i = \varepsilon \Rightarrow \big(B_{i-1} \nsubseteq B_i \text{ and } B_i \nsubseteq B_{i-1}\big)$.
\end{enumerate}

Note that if a sequence $(w_n)_n$ is adequate, then there exists a proper factorization pattern $(\vec u, \vec B)$ such that $(w_n)_n$ is $(\vec u, \vec B)$-adequate. This is easily seen from the following observations and their symmetric counterparts:  
\[
\begin{array}{rcl}
  u = u_1 \cdots u_k \text{ and } u_k \in B &\ \Rightarrow\ & u_1 \cdots u_k B^n \subseteq u_1 \cdots u_{k-1} B^n, \\
  B_{i-1} \subseteq B_i &\ \Rightarrow\ & B_{i-1}^n B_{i}^n \subseteq B_i^n.
\end{array}
\] 

The following lemma gives a condition under which two sequences share a factorization pattern. This lemma is very similar to \cite[Theorem 8.2.6]{JAbook}.

\begin{lemma}
  \label{lem:samepatt}
  Let $(\vec u,\vec B)$ and $(\vec t, \vec C)$ be proper
  factorization patterns.
  Let $(v_n)_{n}$ and  $(w_n)_{n}$  be two sequences of words such that
  \begin{itemize}[leftmargin=10ex]
  \item  $(v_n)_{n}$ is $(\vec u,\vec B)$-adequate

  \item   $(w_n)_{n}$ is $(\vec t, \vec
    C)$-adequate
 \item $v_n \sim_n w_n$  for every $n\geq0$.
  \end{itemize}
  Then, $\vec u=\vec t$ and $\vec B = \vec C$.
\end{lemma}

\begin{proof}
   For a factorization pattern $(\vec u,\vec B)$, we define 
   \[
   \| (\vec u,\vec B) \| := (\sum_{i=0}^{p} |u_i| )+p,
   \] 
   where $p$ is the length of the vector $\vec u$. Let 
   \[
   k := \max(\| (\vec u,\vec B) \|, \| (\vec t,\vec C) \|) +1.
   \] 

  Consider the first word of the sequence $(v_n)_{n}$, \emph{i.e.},~$v_0 = u_0 b_1 u_1 \cdots b_{p} u_p$, where $\content{b_i} = B_i$. Define 
  \[
  v_0^{(k)} := u_0 b_1^k u_1 \ldots b_{p}^k u_p.
  \]
  Recall that $(v_n)_{n}$ being a $(\vec u,\vec B)$-adequate sequence
  means that $$v_n \in \LangExp  u B n p$$ for every $n$. Thus,
  we have for every $\ell \geq k \cdot  \max(|b_1|, \ldots, |b_n|)$ that $v_0^{(k)} \lhd v_\ell$. 
  Note that whenever $\ell ' \geq \max(\ell, | v_0^{(k)}|)$, we have that $v_0^{(k)} \in$ Sub$_{\ell '}(v_{\ell '})$. And, using the assumption that $v_n \sim_n w_n$ for all $n$, this gives that $v_0^{(k)} \lhd w_{\ell '}$.
    In the same way, for $w_0 = t_0 c_1 t_1 \cdots c_{q} t_q$,  we obtain
  an index $m$ such that for every $m'  \geq \max(m, | w_0^{(k)}|)$, both $w_0^{(k)} \lhd w_{m'}$ and $w_0^{(k)}
  \lhd v_{m'}$ hold.   
  
Let $M := \max(\ell ', m')$. Then $v_0^{(k)} \lhd v_{M}, w_{M}$ and $w_0^{(k)} \lhd v_{M}, w_{M}$.

  \smallskip
  Now fix a factor $b_i^k$ of $v_0^{(k)}$. In particular, $b_i^k\lhd
  w_M$. Since $k >  \| (\vec t,\vec C) \|$ and $|b_i| >0$, the pigeonhole principle gives that there is some $C_j$ with $\content{b_i} \subseteq C_j$.

 Exploiting this, we want to define a bijection between the set of indexed alphabets in $\vec B$ and the set of those in $\vec C$ that will help us to show that $(\vec u,\vec B) = (\vec t, \vec C)$.

  Let $\mathbf{B} := \{ (B_1, 1),\ldots, (B_{p}, p)\}$ and $\mathbf{C} := \{
  (C_1,1), \ldots, (C_{q},q)\}$. We define a function $f : \mathbf{B} \rightarrow \mathbf{C}$, by sending $(B_i,i)$ to that $(C_j,j)$ for which $c'_j\in (\exactcontent{C_j})^M$ is the first factor of $w_{M}$ used to fully read $b_i$, while reading $v_0^{(k)}$ as a scattered subword of $w_{M}$.

  The function $g : \mathbf{C} \rightarrow \mathbf{B}$ is defined analogously. The functions $f$ and $g$ preserve the order of the indices and pointwise preserve the alphabet.
  If we show that $f$ and $g$ define a bijective correspondence between $\mathbf{B}$ and $\mathbf{C}$, then $p=q$. The fact that $f$ and $g$ pointwise preserve the alphabet would then imply that $B_i = C_i$, for every $i$. \\

  To establish that $f$ and $g$ are each others inverses, we apply Lemma 8.2.5 from \cite{JAbook}, which we shall first repeat:  
  \begin{lemma}[{\cite[Lemma~8.2.5]{JAbook}}]
    \label{lemma-825}        
    Let $X$ and $Y$ be finite sets and let $P$ be a partially ordered
    set. Let $f : X \rightarrow Y, g : Y \rightarrow X, p : X
    \rightarrow P$ and $q : Y \rightarrow P$ be functions such that
    \begin{enumerate}
    \item\label{it:1} for any $x \in X, p(x) \leq q(f(x))$,
    \item\label{it:2} for any $y \in Y, q(y) \leq p(g(y))$,
    \item \label{it:3}if $x_1, x_2 \in X, f(x_1) = f (x_2)$ and $p(x_1)
      = q ( f(x_1))$, then $x_1 = x_2$,
    \item \label{it:4}if $y_1, y_2 \in Y, g(y_1) = g (y_2)$ and $q(y_1)
      = p ( g(y_1))$, then $y_1 = y_2$.
    \end{enumerate}
    Then $f$ and $g$ are mutually inverse functions and $p = q \circ f$ and $q = p \circ g$.     
  \end{lemma}

   The functions $f$ and $g$ fulfill the conditions of this lemma, if we let $X = \mathbf{B}, Y = \mathbf{C}$, let $P$ be the set of alphabets, partially ordered by inclusion, and let $p$ and $q$ be the projections onto the first coordinate:
     
   \ref{it:1} and~\ref{it:2} hold since $f$ and $g$ pointwise preserve the alphabet. 
  Suppose that $f(B_{i_1}) = f(B_{i_2})$ and that $B_{i_1} = f(B_{i_1}) $. This means that a factor $b_{i_1}$ and a factor $b_{i_2}$ of $v_0^{(k)}$ are read inside the same factor $c'_j$ of $w_{M}$. Thus $\content{b_{i_1} u_{i_1} \cdots {b_{i_2}}} \subseteq \content{c'_j} = f(B_{i_1}) = B_{i_1} = \content{b_{i_1}}$. But we assumed that $(\vec u,\vec B)$ is a \emph{proper} factorization pattern, so $i_1$ must be equal to $i_2$. This shows that~\ref{it:3} holds, and~\ref{it:4} is proved similarly. 

It follows that indeed $f$ and $g$ define a bijective correspondence between $\mathbf{B}$ and $\mathbf{C}$, thus  $p=q$ and $B_i = C_i$, for every $i$. Since we are dealing with proper factorization patterns, $v_0^{(k)} \lhd w_{M}$ now implies that $t_i \lhd u_i$, for every $i$. On the other hand, $w_0^{(k)} \lhd v_{M}$ now implies that $u_i \lhd t_i$, for every $i$. Thus, for every $i$, $u_i = t_i$.
\end{proof}

Now we are equipped to prove the ``only if'' direction of Proposition~\ref{prop:J-sep-criterion}. 

\begin{lemma}
  \label{lem:common-ub-path-2}
  If the languages recognized by two DFAs $\mathcal{A}_1$ and
  $\mathcal{A}_2$ are not \PT-separable, then $\mathcal{A}_1$ and $\mathcal{A}_2$ share a
  common $(\vec u,\vec B)$-path.
\end{lemma}

\begin{proof}
  By hypothesis, for every $n \in \mathbb{N}$, there exist $v_n \in L(\mathcal{A}_1)$ and $w_n \in L(\mathcal{A}_2)$ such that
  \begin{equation}
    \label{eq:1}
    v_n \sim_n w_n.
  \end{equation}
  This defines an infinite sequence of pairs $(v_n,w_n)_n$, from which we will iteratively extract infinite subsequences to obtain additional properties, while keeping~\eqref{eq:1}. 

  By Lemma~\ref{lem:extract-adequate-subseq}, one can extract from $(v_n,w_n)_n$ a subsequence whose first component forms an adequate sequence. From this subsequence of pairs, using Lemma~\ref{lem:extract-adequate-subseq} again, we extract a subsequence whose second component is also adequate (note that the first component remains adequate). Therefore, one can assume that both $(v_n)_n$ and $(w_n)_n$ are themselves adequate.



  Lemma~\ref{lem:samepatt} shows that one can choose the \emph{same} factorization pattern $(\vec u,\vec B)$ such that both $(v_n)_n$ and $(w_n)_n$ are $(\vec u,\vec B)$-adequate. Finally, by the following claim, we then obtain that both $\mathcal{A}_1$ and $\mathcal{A}_2$ admit a $(\vec{u}, \vec{B})$-path. 

\begin{myclaim}
If $L(\mathcal{A})$ contains a  $(\vec{u}, \vec{B})$-adequate sequence, then $\mathcal{A}$ admits a $(\vec{u}, \vec{B})$-path. 
\end{myclaim}

 Indeed, $L(\mathcal{A})$ contains a $(\vec{u}, \vec{B})$-adequate sequence $(v_n)_n$, i.e.
  \[
  \forall n \geq 0,\  v_n \in \LangExp  u B n p \cap L(\mathcal{A}).
  \]
  Let $v_n$ be a sufficiently large term in this sequence, e.g.~with $n > |Q(\mathcal{A})|$. Now the path used to read $v_n$ in $\mathcal{A}$ must traverse loops labeled by each of the $B_i$'s and clearly, by the shape of $v_n$, this is a $(\vec{u}, \vec{B})$-path. 
%
\end{proof}

\newpage\appendix

\def\pv#1{\ensuremath{{\sf#1}}}
\def\tcl#1{\ensuremath{\mathrm{cl}(#1)}}
\let\cl\tcl
\def\Om#1#2{\ensuremath{\overline\Omega_{#1}{\sf#2}}}
\def\om#1#2{\ensuremath{\Omega_{#1}{\sf#2}}}

\section{Connection with profinite semigroup theory:   overview}

We show that separability of two languages by a \pv
V-recognizable language is equivalent to the nonemptiness of the
intersection of their closures in the free pro-\pv V semigroup. This
was already shown in~\cite{MR1709911}. The material of
Section~\ref{sec:preliminary-material} can be found
in~\cite{JAbook}.

\subsection{Background}
\label{sec:preliminary-material}

Fix a finite alphabet $A$ and a pseudovariety~\pv V. A semigroup $T$
\emph{separates} $u,v\in A^+$ if there exists a morphism $\varphi:A^+\to T$
such that $\varphi(u)\neq\varphi(v)$.  Given $u,v\in A^+$, let $r_\pv
V(u,v)=\min\bigl\{|T|:T\in\pv V \text{ and $T$ separates $u$ and
  $v$}\bigr\}\in\mathbb{N}\cup\{\infty\}$. Assume for simplicity that two distinct
words can be separated by some semigroup of \pv V. Then $d_\pv
V(u,v)=2^{-r_\pv V(u,v)}$, with $2^{-\infty}=0$, defines a metric
on~$A^+$. A sequence $(u_n)_n$ is Cauchy for this metric if for every
morphism $\varphi:A^+\to T$, $(\varphi(u_n))_n$ is eventually constant. Let
$(\Om AV,d_\pv V)$ be the completion of the metric space~$(A^+,d_\pv
V)$. By construction, $A^+$ is dense in $\Om AV$. Pointwise
multiplication of classes of Cauchy sequences transfers the semigroup
structure of $A^+$ to~$\Om AV$, on which the multiplication is continuous.

\begin{proposition}
  \label{prop:compact}
  $(\Om AV,d_\pv V)$ is compact.
\end{proposition}

\begin{proof}
  One checks that every sequence $(u_n)_n$ of elements of \Om AV has a
  converging subsequence, that is, since \Om AV is complete, a Cauchy
  subsequence. Since $A^+$ is dense in \Om AV, one can find a word
  $v_n$ such that $\lim_nd_\pv V(u_n,v_n)=0$. This reduces the statement
  to the case where $u_n$ is a word. Now, since there is a finite number of
  morphisms from $A^+$ into a semigroup of size at most $k$, one
  can extract by diagonalization a subsequence $(u'_n)^{}_n$ of
  $(u_n)_n$ such that for any morphism $\varphi:A^+\to T$ with $|T|\leq k$,
  $(\varphi(u'_n))_{n\geq k}$ is constant. 
\end{proof}
\smallskip Endow $T\in\pv V$ with the discrete topology. The definition
of $d_\pv V$ makes every morphism $\varphi:A^+\to T\in\pv V$
uniformly continuous. Since $A^+$ is dense in \Om AV compact, $\varphi$ has a unique
continuous extension $\hat\varphi:\Om AV\to T$ (which by continuity of the
multiplication is also a morphism). For $L\subseteq\Om AV$, denote by
$\tcl L$ its topological closure in \Om AV.

\begin{lemma}
  \label{lem:cl}
  Let $\varphi:A^+\to T\in\pv V$ and $K=\varphi^{-1}(P)$ for $P\subseteq T$. Then $\tcl K=\hat\varphi^{-1}(P)$.
\end{lemma}

\begin{proof}
  Unions commute with inverse images and closures, so it suffices to
  treat the case $P=\{p\}$. Since $\hat\varphi$ is continuous, $\hat\varphi^{-1}(p)$
  is clopen, and it contains $K$, so $\tcl
  K\subseteq\hat\varphi^{-1}(p)$. Conversely, for $u\in\hat\varphi^{-1}(p)$, pick a word
  $u_n$ such that $d_\pv V(u,u_n)<2^{-n}$ (which exists since $A^+$ is dense
  in \Om AV). Then $\varphi(u_n)=p$ for $n>|T|$, hence $u_n\in K$,
  so~$u\in\cl{K}$.
\end{proof}

\noindent For $K\subseteq A^+$, we let $K^c=A^+\setminus K$ and $\cl K^c=\Om AV\setminus\cl K$.
\begin{corollary}
  \label{cor:cl}
  \begin{enumerate}
  \item\label{item:1} If $K$ is \pv V-recognizable, then $\tcl {K^c}=\tcl{K}^c$.
  \item\label{item:2} If $K$ is \pv V-recognizable and $L\subseteq A^+$ is such that
    $\tcl{L}\subseteq\tcl{K}$, then $L\subseteq K$.
  \end{enumerate}
\end{corollary}

\begin{proof}
  \ref{item:1} Let $\varphi:A^+\to T\in\pv V$, with $K=\varphi^{-1}(P)$. By
  Lemma~\ref{lem:cl}, $\cl{K^c}=\hat\varphi^{-1}(T\setminus P)=\Om AV\setminus\hat\varphi^{-1}(P)=\cl{K}^c$.  For \ref{item:2},
  just write $L\cap K^c\subseteq\tcl{K}\cap\tcl{K^c}=\emptyset$ by~\ref{item:1}.
\end{proof}

\begin{proposition}[follows from {\cite[Thm.~3.6.1]{JAbook}}]
  \label{prop:basis}
  Closures of\/ \pv V-recognizable languages form a basis of the
  topology of\/ \Om AV.
\end{proposition}

\begin{proof}
  By Lemma~\ref{lem:cl}, the closure of a \pv V-recognizable language
  is of the form $\hat\varphi^{-1}(P)$ for some continuous morphism
  $\hat\varphi:\Om AV\to T\in\pv V$, hence it is open. Conversely, for $u\in\Om
  AV$, let $O_u=\hat\alpha^{-1}(\hat\alpha(u))$, where $\alpha$ is the product of all
  morphisms $\varphi:A^+\to T\in\pv V$ for $|T|\leq n$. By Lemma~\ref{lem:cl},
  $O_u$ is the closure of the \pv V-recognizable language
  $\alpha^{-1}(\hat\alpha(u))$. By construction, $O_u$ is an open containing $u$, contained
  in the ball of radius $2^{-n}$ centered at~$u$.
\end{proof}

\subsection{\texorpdfstring{Separability of languages by a \pv
    V-recognizable language}{Separability of languages by a
    V-recognizable language}}
\label{sec:separability-pv-v}

Two languages $L_1,L_2\subseteq A^+$ are \emph{\pv V-separable} if there
exists a \pv V-recognizable language $K$ such that $L_1\subseteq K$ and $K\cap
L_2=\emptyset$. Such a language $K$ is a witness, in the given variety of
languages, that~$L_1\cap L_2=\emptyset$, and we say that it \emph{separates}
$L_1$ and $L_2$.

\begin{proposition}
  \label{prop:separation}
  Two languages of~$A^+$ are separated by a \pv V-recognizable
  language if and only if the intersection of their topological
  closures in \Om AV is~empty.
\end{proposition}

\begin{proof}
  Let $L_1,L_2\subseteq A^+$, and let $K$ be \pv V-recognizable such that
  $L_1\subseteq K$ and $K\cap L_2=\emptyset$. Then
  $\tcl{L_1}\cap\cl{L_2}\subseteq\tcl{K}\cap\tcl{K^c}=\emptyset$ by Corollary~\ref{cor:cl}.

  Conversely, if $\tcl{L_1}\cap\tcl{L_2}=\emptyset$, then any $u\in\cl{L_1}$
  belongs to the open set $\cl{L_2}^c$, so by
  Proposition~\ref{prop:basis}, there exists some \pv V-recognizable
  language $K_u$ whose closure $O_u$ contains $u$, and is such that
  $O_u\cap\cl{L_2}=\emptyset$. Therefore $\cl{L_1}\subseteq\bigcup_{u\in\cl{L_1}}O_u$. Since $\cl{L_1}$
 is a closed set in the compact space \Om AV (Prop.~\ref{prop:compact}), it is itself
  compact and has a finite cover~$O_{u_1}\cup\cdots\cup O_{u_n}$. Then $K=K_{u_1}\cup\cdots\cup K_{u_n}$ is \pv
  V-recognizable. We have $\cl{L_1}\subseteq\cl{K}$, so by
  Corollary~\ref{cor:cl},~${L_1}\subseteq{K}$.  
  Also, $K \subseteq O_{u_1}\cup\cdots\cup O_{u_n} \subseteq \cl{L_2}^c$. 
\end{proof}

\section{Proof of Lemma~\ref{lem:sat red}}
\setcounter{lemma}{0} 

\begin{lemma}
  The following problem is NP-complete. 
  \begin{tabbing}
    \textsf{Input:} \hspace{4mm} \= An alphabet $A = \{a_1, a_2, \ldots, a_n\}$ and two DFA's $\mathcal{A}_1, \mathcal{A}_2$ over $A$.\\
    
    \textsf{Question:} \> Do there exist $u \in L(\mathcal{A}_1)$ and $v \in L(\mathcal{A}_2)$ such that $\content{u} = \content{v}$?
  \end{tabbing}
\end{lemma}

\begin{proof}
  We will give a reduction from $3$-SAT to this problem. 

  Let $\varphi$ be a $3$-SAT formula over the variables $\{x_1, \ldots, x_n\}$. Define $A := \{x_1, \ldots x_n, \neg x_1,\ldots, \neg x_n\}$. 
  Let $\mathcal{A}_1$ be 

  
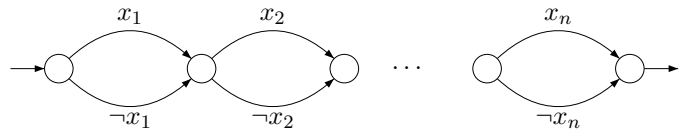
\begin{figure}[H]
    \begin{center}
      \unitlength=4pt
        \scalebox{.9}{
      \begin{gpicture}(60, 6)(0,-3)
        \gasset{Nw=3,Nh=3,Nmr=2.5,curvedepth=0}
        \thinlines
        \node[Nmarks=i,iangle=180](A1)(0,0){}
        
        \node(A2)(15,0){}
        \node(A3)(30,0){}
        \node(A4)(45,0){}

        \node[Nmarks=f,iangle=180](A5)(60,0){}
        \gasset{curvedepth=4}
        \drawedge(A1,A2){$x_1$}
        \drawedge(A2,A3){$x_2$}
        \put(35,0){$\ldots$}
        \drawedge(A4,A5){$x_n$}
        \gasset{curvedepth=-4}
        \drawedge[ELside=r](A1,A2){$\neg x_1$}
        \drawedge[ELside=r](A2,A3){$\neg x_2$}
        \drawedge[ELside=r](A4,A5){$\neg x_n$}
      \end{gpicture}}
    \end{center}
  \end{figure}
  and let $\mathcal{A}_2$ be the serial automaton in which for every disjunct $d$ in the $i$-th clause of $\varphi$, there is an arrow from state $i$ to $i+1$ labeled $d$, concatenated with a copy of $\mathcal{A}_1$. For example, if $\varphi = (x_1 \vee x_3 \vee \neg x_4) \wedge \ldots \wedge (x_4 \vee \neg x_5 \vee x_2)$, the automaton $\mathcal{A}_2$ is

  
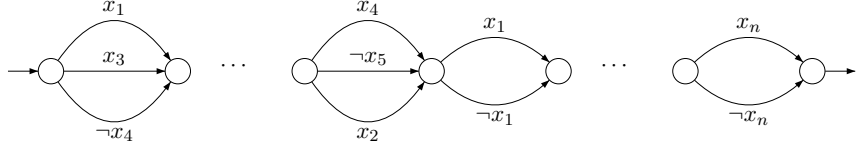
\begin{figure}[H]
    \begin{center}
      \unitlength=4pt
        \scalebox{.8}{
      \begin{gpicture}(100, 8)(0,-3)
        \gasset{Nw=3,Nh=3,Nmr=2.5,curvedepth=0}
        \thinlines
        \node[Nmarks=i,iangle=180](A1)(0,0){}
        
        \node(A2)(15,0){}
        \node(A3)(30,0){}
        \node(A4)(45,0){}
        \node(A5)(60,0){}
        \node(A6)(75,0){}
        \node[Nmarks=f,iangle=180](A7)(90,0){}
        \put(20,0){$\cdots$}
        \put(65,0){$\cdots$}
        
        \gasset{curvedepth=0}
        \drawedge(A1,A2){$x_3$}
        \drawedge(A3,A4){$\neg x_5$}

        \gasset{curvedepth=6}
        \drawedge(A1,A2){$x_1$}  
        \drawedge(A3,A4){$x_4$}
        \gasset{curvedepth=4}
        \drawedge(A4,A5){$x_1$}
        \drawedge(A6,A7){$x_n$}

        \gasset{curvedepth=-6}
        \drawedge[ELside=r](A1,A2){$\neg x_4$}
        \drawedge[ELside=r](A3,A4){$x_2$}
        \gasset{curvedepth=-4}
        \drawedge[ELside=r](A4,A5){$\neg x_1$}
        \drawedge[ELside=r](A6,A7){$\neg x_n$}      
      \end{gpicture}}
    \end{center}
  \end{figure}
  
  We will show that $\varphi$ is satisfiable if and only if the question mentioned above is answered positively for these $\mathcal{A}_1$ and $\mathcal{A}_2$.

  Suppose $\varphi$ is satisfiable. Then there is a valuation $v : \{x_1, \ldots x_n\} \rightarrow \{0,1\}$ such that $\overline{v}(\varphi) = 1$. Define $u := y_1 \cdots y_n$, with $y_i = x_i$ if $v(x_i)=1$ and $y_i = \neg x_i$ if $v(x_i) = 0$. In each of the $k$ clauses of $\varphi$, there is at least one disjunct $d$ for which $v(d) = 1$.  
  Define $v := w_1 \cdots w_k u$, where $w_i$ is any one of the disjuncts in the $i$-th clause that is evaluated to $1$. 
  Now, $u \in L(\mathcal{A}_1), v \in L(\mathcal{A}_2)$, and by soundness of the valuation function, $\content{u} = \content{v}$. 

  On the other hand, suppose that for these $\mathcal{A}_1$ and $\mathcal{A}_2$, there are $u \in L(\mathcal{A}_1), v \in L(\mathcal{A}_2)$ with $\content{u} = \content{v}$.
  By construction of $\mathcal{A}_1$, for every $i$, $\content{u}$ contains either $x_i$ or $\neg x_i$. 
  By construction of $\mathcal{A}_2$ and since $\content{u} = \content{v}$, we have that $v = wu$ and that $\content{w} \subseteq \content{u}$. 
  Define the valuation 
  \[
  \begin{array}{rcccl}
    v : & \{x_1, \ldots x_n\} &\rightarrow& \{0,1\} & \\
    & x_i & \mapsto & 1& \text{if } x_i \in \content{u} \\
    & x_i & \mapsto & 0& \text{else}
  \end{array}
  \]
  Now $v$ sends all variables occurring in $w$ to $1$, which gives $\overline{v}(\varphi) =1$.

\end{proof}

\end{document}